\documentclass[11pt]{article}

\usepackage{amsthm,amsmath,color,graphicx, fullpage}
\usepackage{boxedminipage}
\usepackage[noend]{algorithmic}
\usepackage[section,boxed]{algorithm}

\DeclareGraphicsExtensions{.eps}

\newtheorem{lemma}{Lemma}
\newtheorem{corr}{Corollary}
\newtheorem{theorem}{Theorem}

\newcommand{\degree}{\mathrm{degree}}

\newcommand{\G}{\mathrm{G}}
\newcommand{\V}{\mathrm{V}}
\newcommand{\E}{\mathrm{E}}

\newcommand{\ID}{\mathrm{ID}}
\newcommand{\Xh}{\emph{Xheal\ }}

\begin{document}
%
% --- Author Metadata here ---
%\conferenceinfo{Arxiv}{\date{04/05/2011}} 
%\CopyrightYear{2011} 
%\crdata{978-1-4503-0719-2/11/06} 
\clubpenalty=10000 
\widowpenalty = 10000

%\conferenceinfo{WOODSTOCK}{'97 El Paso, Texas USA}
%\CopyrightYear{2007} % Allows default copyright year (20XX) to be over-ridden - IF NEED BE.
%\crdata{0-12345-67-8/90/01}  % Allows default copyright data (0-89791-88-6/97/05) to be over-ridden - IF NEED BE.
% --- End of Author Metadata ---

\title{Xheal: Localized Self-healing using Expanders \thanks{A shorter version of this paper published at PODC 2011, San Jose, CA.}}

\author{Gopal Pandurangan \thanks{Division of Mathematical
Sciences, Nanyang Technological University, Singapore 637371 and Department of Computer Science, Brown University, Providence, RI 02912.  \hbox{E-mail}:~{\tt gopalpandurangan@gmail.com}. Supported in part by NSF grant CCF-1023166 and by a grant from the United States-Israel Binational Science Foundation (BSF).}   \and Amitabh Trehan \thanks{Faculty of Industrial Engineering and Management, Technion - Israel Institute of Technology, Haifa, Israel - 32000. \hbox{E-mail}:~{\tt amitabh.trehaan@gmail.com.} 
Work done partly at  Brown University and University of Victoria. Supported in part at the Technion by a fellowship of the Israel Council for Higher Education.
}}

\maketitle
\begin{abstract}
We consider the problem of self-healing in reconfigurable networks (e.g. peer-to-peer and wireless mesh networks) that are under repeated attack by an omniscient adversary and propose a fully distributed algorithm, \Xh, that maintains good expansion and spectral properties of the network, also keeping the network connected. Moreover, \Xh, does this while allowing only low stretch and degree increase per node. The algorithm heals global properties like expansion and stretch while only doing local changes and using only local information.
 We use a model similar to that used in recent work on self-healing.  In our model, over a sequence of rounds, an adversary either inserts a node with arbitrary  connections or deletes an arbitrary node from the network. The network responds  by quick ``repairs," which consist of adding or deleting edges in an efficient localized manner.

These repairs  preserve the edge expansion, spectral gap, and network stretch, after adversarial deletions, without increasing node degrees by too much, in the following sense.   At any point in the algorithm, the expansion of the graph will be either `better' than the expansion of  the graph formed by considering only the adversarial insertions (not the adversarial deletions) or the expansion will be, at least, a constant.   Also, the stretch i.e.  the distance between any pair of  nodes in the healed graph is  no more than a $O(\log n)$ factor.  Similarly, at any point, a node $v$ whose degree would have been $d$ in the graph with adversarial insertions only, will have degree at most $O(\kappa d)$ in the actual graph, for a small parameter $\kappa$.  We also  provide bounds on the second smallest eigenvalue of the Laplacian which captures key properties such
as mixing time, conductance, congestion in routing etc. Our distributed data structure  has low amortized latency and bandwidth requirements.
Our work improves over the self-healing algorithms \emph{Forgiving tree} [PODC 2008] and \emph{Forgiving graph} [PODC 2009]  in that we are able to give   guarantees on degree and stretch, while at the same time preserving the expansion and spectral properties of the network.

\iffalse

We consider the problem of self-healing in peer-to-peer networks that are under repeated attack by an omniscient adversary. We assume that, over a sequence of rounds, an adversary either inserts a node with arbitrary connections or deletes an arbitrary node from the network. The network responds to each such change by quick ``repairs," which consist of adding or deleting edges in an efficient localized manner.

These repairs  preserve the edge expansion after adversarial deletions, without increasing node degrees by too much, in the following sense.   At any point in the algorithm, the expansion of the graph will be `better' than the expansion of  the graph formed by considering only the adversarial insertions (not the adversarial deletions).  Similarly, at any point, a node $v$ whose degree would have been $d$ in the graph with adversarial insertions only, will have degree at most $kd$ in the actual graph, for a constant $k$. Our distributed data structure 
has low latency and bandwidth requirements. At a high level, each node maintains  an `expander cloud' of virtual nodes simulated by its neighbors. 
Our work improves over the self-healing algorithms of Forgiving tree [PODC 2008] and Forgiving graph [PODC 2009] papers in that we are able to give   guarantees on degree and stretch, while at the same time preserving the expansion and spectral property of the network.

\fi

\end{abstract}

%\category{C.2.1}{Computer-Communication Networks} {Network Architecture and Design;} {Distributed networks, Network communications, Network topology, Wireless communication}
%\category{C.2.4}{Computer-Communication Networks} {Distributed Systems}
%\category{C.4}{Computer Systems Organization} {Performance of Systems} {Fault tolerance, Reliability, availability and serviceability}
%\category{E.1}{Data Structures} {Distributed data structures, Graphs and networks}
%\category{G.2.2}{Graph Theory} {Graph algorithms}
%\category{G.3}{Probability and Statistics} {Probabilistic algorithms}
%\category{H.3.4}{Systems and Software} {Distributed systems, Information networks}

%\terms{ Algorithms, Design, Reliability, Security, Theory}

%\keywords{self-healing, reconfiguration, local, distributed, expansion, spectral properties, expanders, randomized}

\section{Introduction}

Networks in the modern age have grown to such an extent that they have now begun to resemble self-governing living entities. Centralized control and management of resources has become increasingly untenable. Distributed and  localized attainment of self-* properties is fast becoming the need of the hour. 

As we have seen the baby Internet grow through its adolescence into a strapping teenager, we have experienced and are experiencing many of its growth pangs and tantrums. There have been recent disruption of services in networks such as Google, Twitter, Facebook and Skype.  On August 15, 2007 the Skype network crashed for about $48$ hours, disrupting service to approximately $200$ million users~\cite{fisher,malik, moore, ray, stone}.  Skype attributed this outage to failures in their ``self-healing mechanisms''~\cite{garvey}.  We believe that this outage is indicative of  the  unprecedented complexity of modern computer systems: we are approaching scales of billions of components.  Unfortunately, current algorithms ensure robustness in computer networks through the increasingly unscalable approach of hardening individual components or, at best, adding lots of redundant components. Such designs are increasingly unviable. No living organism is designed such that no component of it ever fails: there are simply too many components. For example, skin can be cut and still heal. It is much more practical to  design skin that can heal than a skin that is completely impervious to attack. 

This paper adopts a \emph{responsive} approach, in the sense that it responds to an attack (or component failure) by changing the topology of the network.  This  approach works irrespective of the initial state of the network, and is thus orthogonal and complementary to traditional  non-responsive techniques. This approach requires the network to be \emph{reconfigurable}, in the sense that the topology of the network can be changed. 
Many important networks are \emph{reconfigurable}. Many of these  we have designed e.g.  peer-to-peer, wireless mesh and ad-hoc computer networks, and  infrastructure networks, such as an airline's transportation network. Many have existed since long but we have only now closely scrutinized them e.g. social networks such as friendship networks on social networking sites, and  biological networks, including the human brain. Most of them are also dynamic, due to the capacity of  individual nodes to initiate new connections or drop existing connections.

 In this setting, our paper seeks to address the important and challenging problem of efficiently and responsively maintaining global invariants in a localized, distributed manner.  It is obvious that it is a significant challenge to come up with approaches to optimize various properties at the same time, especially with only local knowledge. For example, a star topology achieves the lowest distance between nodes, but the central node has the highest degree. If we were trying to give the lowest degrees to the nodes in a connected graph, they would be connected in a line/cycle giving the maximum possible diameter. Tree structures give a good compromise between degree increase and distances, but may lead to poor spectral properties (expansion) and poor load balancing. 
Our main contribution is a self-healing algorithm \Xh that maintains  spectral properties (expansion), connectivity, and stretch in a distributed manner using only localized information and actions, while allowing only a small degree increase per node. %\Amitabh{What about lower bound? -> We also show these bounds are asymptotically tight.}
  Our main algorithm is described in Section~\ref{sec:algorithm}.

%general discussion: expansion/ expansion healing

 %\begin{observation}
% \begin{lemma}
% The algorithm $\FG$ does not maintain expansion while self-healing.
% \end{lemma}
% 
% \begin{proof}
% Consider a star with $2^y + 1$ nodes with central node $v$. Let $v$ be the node deleted by the adversary. The healed graph $H$ will be a tree with the $2^y$ neighbors of $v$ as the leaf nodes and $2^y - 1$ virtual internal nodes simulated by these leaves (at most one virtual node per leaf) such that this tree is a binary search tree if we consider the leaf nodes to be sorted in ascending order from left to right. The real network $G$ is an isomorphism of $H$ with the virtual nodes mapped to the nodes simulating them. Thus, we may hope for a better expansion. However, consider the  root node $H$, say $r$, and its right child, say $x$. Since the nodes simulating them are in separate subtrees of this BST, the edge $(r,x)$ will form a cut in $G$. Thus, edge expansion $\h_G = 2^{1-y}$ as compared to $1$ in the original graph.
% \end{proof}

\medskip
%copied from FG
\noindent {\bf Our Model:}
Our model, which is similar to the model introduced in  \cite{HayesPODC09, Amitabh-2010-PhdThesis}, is briefly described here.  We assume that the network is initially a connected (undirected, simple) graph over $n$ nodes.  An adversary repeatedly attacks the
network. This adversary knows the network topology and our algorithm, and it has the ability to delete arbitrary nodes from the  network or insert a new node in the system which it can connect to any subset of  nodes currently in the system.  
However, we assume the adversary is constrained in that in any time step it can only delete or insert a single node. (Our algorithm can be extended to handle
multiple insertions/deletions.) The detailed model is described in Section~\ref{sec: Xmodel}.

\medskip
\noindent {\bf Our Results:} 
  For a reconfigurable network (e.g., peer-to-peer, wireless mesh networks)  that has both insertions and deletions, let $G'$ be the graph consisting of the original nodes and inserted nodes without any changes due to deletions. Let $n$ be the number of nodes in $G'$, and $G$ be the present (healed) graph. Our main result is a new algorithm \Xh that ensures (cf. Theorem \ref{th:main} in Section \ref{sec: Results}): 
 1) \emph{Spectral Properties:} If $G'$ has expansion equal or better than a constant,  \Xh achieves at least a constant expansion, else it maintains at least the same expansion as $G'$;  Furthermore, we show bounds on the second smallest eigenvalue of the Laplacian of $G$, $\lambda(G)$ with respect to
 the corresponding $\lambda(G')$. An important special case of our result
 is that if $G'$ is an (bounded degree) expander, then Xheal guarantees that $G$ is also an (bounded degree) expander.  We note that such a guarantee is not provided by the self-healing algorithms of \cite{HayesPODC09, HayesPODC08}.
   2)\emph{Stretch:} The distance between any two nodes of the actual network never increases by more than $O(\log n)$ times their distance in $G'$; and 3) the degree of any node never increases by more than $\kappa$ times its degree in $G'$, where $\kappa$ is a small parameter (which is implementation dependent, can be chosen to be a constant --- cf. Section \ref{sec: distributed}).

    Our algorithm is  distributed, localized and resource efficient.  
We introduce the main algorithm separately (Section~\ref{sec:algorithm}) and a distributed implementation (Section~\ref{sec: distributed}). The high-level idea
behind our algorithm is to put a $\kappa$-regular expander between the deleted node and its neighbors. Since this expander has low degree and constant expansion, intuitively this helps in maintaining good expansion. However, a key
complication in this intuitive approach is efficient implementation while maintaining bounds on degree and stretch.  The $\kappa$ parameter above is determined by the particular distributed implementation of an expander that we use. Our construction is randomized which guarantees efficient maintenance
of an expander under insertion and deletion, albeit at the cost of a small probability that the graph may not be an expander. This aspect of our implementation can be improved if one can design efficient distributed constructions that yield expanders deterministically. (To the best of our knowledge no such construction is known).
 In our  implementation, for a deletion, repair takes  $O(\log n)$ rounds and has  amortized complexity that is within $O(\kappa\log n)$ times the best possible. The formal statement and proof of these results are in Sections~\ref{sec: Results} and \ref{sec: distributed}. 

%\Amitabh{Placeholder: Our algorithm is optimal for expansion with constant deg increase/ lower bound?}\\
%The formal statement and proof of these results is in Section~\ref{subsec: upperbounds}. 

\noindent {\bf Related Work:} 
%\Amitabh{Fix the first paras}
The work most closely related to ours is~\cite{HayesPODC09, Amitabh-2010-PhdThesis}, which introduces a distributed data structure \emph{Forgiving Graph} that, in a model similar to ours, maintains low stretch of a network with constant multiplicative degree increase per node.  However, \Xh is  more ambitious in that it not only maintains similar properties but also the spectral properties (expansion) with obvious benefits, and also uses different techniques. However, we pay with larger message sizes and amortized analysis of costs. The works of \cite{HayesPODC09, Amitabh-2010-PhdThesis} themselves use models or techniques from earlier work~\cite{Amitabh-2010-PhdThesis, HayesPODC08, SaiaTrehanIPDPS08, BomanSAS06}. They put in tree like structures of nodes in place of the deleted node. Methods which put in  tree like sructures of nodes are likely to be bad for expansion. If the original network is a star of $n+1$ nodes and the central node gets deleted, the repair algorithm puts in a tree, pulling the expansion down from a constant to $O(1/n)$.  Even the algorithms {\em Forgiving tree} \cite{HayesPODC08} and {\em Forgiving graph}~\cite{HayesPODC09}, which put in a tree of virtual nodes (simulated by real nodes) in place of a deleted node don't improve the situation.  In these algorithms, even though the real network is an isomorphism of the virtual network, the `binary search' properties of the virtual trees ensure a poor cut involving the root of the trees.

%\Amitabh{ Put in previous work on expanders/ maintainance}
The importance of spectral properties is well known~\cite{ chungbook, Wigderson-exsurvey}. Many results are based on graphs having enough expansion or conductance, including recent results in distributed computing in information spreading etc.~\cite{HillelPODC10}.  There are only a few papers showing distributed construction of expander graphs~\cite{lawsiu, spanders, mihail-p2p};  Law and Siu's construction gives  expanders with high probability using Hamilton cycles which we use in our implementation.

 Many papers have  discussed strategies for adding additional capacity or rerouting in
anticipation of failures~\cite{AwerbuchAdapt92, doverspike94capacity, frisanco97capacity, iraschko98capacity, murakami97comparative, caenegem97capacity, xiong99restore}.  Some other results  are also responsive  in some sense: \cite{medard99redundant, anderson01RON} or  have enough built-in redundancy in separate components ~\cite{goel04resilient}, but all of them have fixed network topologies.  Our approach does not dictate routing paths or require initially placed redundant components.  There is also some research in the physics community on preventing cascading failures  which empirically works well  but  unfortunately performs very poorly under adversarial attack~\cite{holme-2002-65, motter-2002-66, motter-2004-93,hayashi2005}.

% Some research in the physics community on preventing cascading failures  have  each vertex in the network starting with a fixed capacity. Deletion of a node and `load' redistribution can cause cascadiing failure.
%   In~\cite{holme-2002-65, motter-2002-66}, the authors  show empirically that even a single node deletion can cause a
%constant fraction of the nodes to fail in a power-law network due to cascading failures. Motter and Lai
%propose a strategy for addressing this problem by intentional removal of certain nodes in the network after a failure begins
%~\cite{motter-2004-93}.  Hayashi and Miyazaki prevent cascading failure by  emergent rewirings (edge insertions)~\cite{hayashi2005}.  These approaches empirically  work well  but  unfortunately perform very poorly under adversarial attack.
%
\subsection{Preliminaries}
%We first define edge expansion in graphs.
 \textbf{Edge Expansion:} 
Let $\G =(\V,\E)$ be an undirected graph and $S \subset \V$ be a set of nodes.
We denote $ \overline{S} = V - S$. Let $|\E|_{S, \overline{S}} = \{ (u,v) \in \E | u \in S, v \in \overline{S} \}$ be the number of edges crossing the cut $(S, \overline{S})$.  We define the {\em volume} of $S$ to be the sum of the degrees
of the vertices in $S$ as $vol(S) = \sum_{x \in S}degree(x)$.% Then the conductance $\h(S, \overline{S})$ is defined as
%\[
% \Phi(S, \overline{S}) = \frac{|\E|_{S,\overline{S}}}{|\E|_{S}}
%\] 
% where $ \overline{S} = V - S, \E_S = \{ (u,v) \in \E | u,v \in S \}\ \mathrm{and}\ \E_{S, \overline{S}} = \{ (u,v) \in \E | u \in S, v \in \overline{S}  \} $.\\
 The edge expansion of the graph $h_G$ is defined as, 
 $
  h_G =  {\mathrm{min}_{ |S| \le |\V|/2}}  \frac{|\E|_{S,\overline{S}}}{|S|}   
 $. 
%The edge expansion of the graph $\h_G$ is defined as, 
% $
%  \h_G =  {\mathrm{min}_{ |S| \le |\V|/2}}  \frac{|\E|_{S,\overline{S}}}{vol(S)}   
% $, 

\noindent  \textbf{Cheeger constant:}  A related notion is the Cheeger constant $\phi_G$ of a graph (also called
 as {\em conductance}) defined as follows \cite{chungbook}:
 $\phi_G = {\mathrm{min}_{|S|}} \frac{|\E|_{S,\overline{S}}}
 {min(vol(S), vol(\overline{S}))}.$
 
 The Cheeger constant can be more appropriate for graphs which are very non-regular, since the denominator takes into account the sum of the degrees of  vertices in $S$, rather than just the size of $S$. Note for $k-$regular graphs,
 the Cheeger constant is just the  edge expansion divided by $k$, hence they
 are essentially equivalent for regular graphs. However, in general graphs, 
 key properties such as  mixing time, congestion in routing etc\. are captured more accurately by the Cheeger constant, rather than edge expansion. For example, consider a constant
 degree expander of $n$ nodes and partition the vertex set into two equal parts. Make each of the parts a clique. This graph has expansion at least a constant, but its conductance is $O(1/n)$. Thus while the expander has logarithmic mixing time, the modified graph has polynomial mixing time.
 
 The Cheeger constant is closely related to the the second-smallest eigenvalue of the  Laplacian matrix denoted
by $\lambda_G$ (also called the ``algebraic connectivity" of the graph).
Hence $\lambda_G$, like the Cheeger constant, captures many key ``global" properties of the graph \cite{chungbook}. $\lambda_G$ 
captures how ``well-connected" the graph is and is strictly greater than 0 (which is always the smallest eigenvalue) if and only if the graph is connected.
For an expander graph, it is a constant (bounded away from zero). The larger $\lambda_G$ is, larger is the expansion.

\begin{theorem}\textsc{Cheeger inequality}\cite{chungbook}
$2\phi_G \geq \lambda_G  > \phi_G^2/2$
\end{theorem}

\section{Node Insert, Delete, and Network Repair Model}
\label{sec: Xmodel}

%\floatname{algorithm}{Model}

%\begin{algorithm}[h!]
\begin{figure}[h!]
\caption{The Node Insert, Delete and Network Repair Model -- Distributed View.}
\label{algo:model-2}
\begin{boxedminipage}{0.9\textwidth}
\begin{algorithmic}
\STATE Each node of $G_0$ is a processor.  
\STATE Each processor starts with a list of its neighbors in $G_0$.
\STATE Pre-processing: Processors may send messages to and from
their neighbors.
\FOR {$t := 1$ to $T$}
\STATE Adversary deletes or inserts a node $v_t$ from/into $G_{t-1}$, forming $U_t$.
\IF{node $v_t$ is inserted} 
\STATE The new neighbors of $v_t$ may update their information and send messages to and from
their neighbors.
\ENDIF
\IF{node $v_t$ is deleted} 
\STATE All neighbors of $v_t$ are informed of the deletion.
\STATE {\bf Recovery phase:}
\STATE Nodes of $U_t$ may communicate (synchronously, in parallel) 
with their immediate neighbors.  These messages are never lost or
corrupted, and may contain the names of other vertices.
\STATE During this phase, each node may insert edges
joining it to any other nodes as desired. 
Nodes may also drop edges from previous rounds if no longer required.
\ENDIF
\STATE At the end of this phase, we call the graph $G_t$.
\ENDFOR
\vspace{10pt}
\hrule
\STATE
\STATE {\bf Success metrics:} Minimize the following ``complexity'' measures:\\
Consider the graph  $G'_t$ which is the graph, at timestep $t$, consisting solely of the original nodes (from $G_0$) and insertions without regard to deletions and healings. 
%Graph $G'_{t}$ is $G'$ at timestep $t$ (i.e. after the $t^{\mathrm{th}}$ insertion or deletion).
%Graph $G'_{t}$ is $G'$ at timestep $t$ which is equivalent to $G'_{t'}$ where the $t' \le t$ is
%the timestep at which the latest insertion on or before $t$ occured. 
\begin{enumerate}
\item{\bf Degree increase.}  $\max_{v \in G_t} \frac{\degree(v,G_t)}{ \degree(v,G'_t)}$
\item {\bf Edge expansion.} $h(G_t) \ge min(\alpha,\beta h(G'_t))$; for constants $ \alpha , \beta > 0$
\item {\bf Network stretch.} $\max_{x, y \in G_{t}} \frac{dist(x,y,G_{t})}{dist(x,y,G'_{t})}$, where, for a graph $G$ and nodes $x$ and $y$ in $G$, $dist(x,y,G)$ is the
length of the shortest path between $x$ and $y$ in $G$.
%\item{\bf Communication per node.} The maximum number of bits sent by a single node in a single recovery round.
% \tom{Want to modify this or omit?}
\item{\bf Recovery time.} The maximum total time for a recovery round,
assuming it takes a message no more than $1$ time unit to traverse any edge and we have unlimited local computational power at each node. We  assume the
LOCAL message-passing model, i.e., there is no bound on the size of the message
that can pass through an edge in a time step.
\item{\bf Communication complexity.} Amortized number of messages used for recovery.

\end{enumerate}
\end{algorithmic}
\end{boxedminipage}
\end{figure}
%\end{algorithm}

%\Amitabh{This may need modification/ editing}

This model is based on the one introduced in~\cite{HayesPODC09, Amitabh-2010-PhdThesis}. Somewhat similar models were also used in~\cite{HayesPODC08, SaiaTrehanIPDPS08}.
We now describe the details.  Let $G = G_0$ be an arbitrary graph on $n$ nodes, which represent processors in a distributed network.  In each step, the adversary either adds a node or deletes a node.  After each deletion, the algorithm gets to add some new edges to the graph, as well as deleting old ones.  At each insertion, the processors follow a protocol to update their information.
The algorithm's goal is to maintain connectivity in the network, while maintaining good expansion properties and  keeping the distance between the nodes small.  At the same time, the algorithm wants to minimize the resources spent on this task, including keeping node degree small.  We assume that although
the adversary has full knowledge of the topology at every step and can add or delete any node it wants, it is oblivious to the random choices made by the self-healing algorithm as well as to the communication that takes place between the nodes
(in other words, we assume private channels between nodes).

%We seek an algorithm
%which gives performance guarantees under these metrics for each of the  possible insertion and deletion orders.

Initially, each processor only knows its neighbors in $G_0$, and is unaware of the structure of the rest of $G_0$.
After each deletion or insertion, only the neighbors of the deleted or inserted vertex are informed that
the deletion or insertion has occurred. After this, processors are allowed to communicate (synchronously) by sending a limited number
of messages to their direct  neighbors.  We assume that these messages are always sent and received successfully.  The
processors may also request new edges be added to the graph.
We assume that no
other  vertex is deleted or inserted until the end of this round of computation and communication has concluded.
%To make this assumption more reasonable, the per-node communication cost should be very small in $n$ (e.g. at most logarithmic).
%$O(1)$ bits, and should 
%moreover be parallelizable so that the entire protocol can be completed in $O(1)$ time. 

We also allow a certain amount of pre-processing to be done before the first attack occurs. In particular, we assume that all nodes have access to some amount of local information. For example, we assume that all nodes know the address
of all the neighbors of its neighbors (NoN). More generally, we assume the (synchronous) \cal{LOCAL} computation model \cite{peleg} for our analysis.
This is a well studied distributed computing model and has been used to study numerous ``local" problems  such as coloring, dominating set, vertex cover etc. \cite{peleg}.
This model allows arbitrary sized messages to go through an edge per time step.
In this model the NoN information can be exchanged in $O(1)$ rounds.

 Our goal is to minimize the time (the number of rounds) and the (amortized) message complexity per deletion (insertion
doesn't require any work from the self-healing algorithm).
 Our model is summarized  in Figure~\ref{algo:model-2}.

\section{The algorithm}
\label{sec:algorithm}

We give a high-level view of the distributed algorithm deferring the distributed implementation details for now (these will be described later in Section~\ref{sec: distributed}). 
The algorithm is summarized in Algorithm ~\ref{algo: repairbyexpander}.
To describe the algorithm, we associate a color  
 with each edge of the graph.  We will assume that the original edges of $\G$ and those added by the adversary are all colored {\bf black} initially. The algorithm can later recolor edges (i.e., to a color other than black --- throughout when we say ``colored" edge we mean a color other than black) as described below. If $(u,v)$ is a black (colored) edge, we say that $v$($u$) is a black (colored) neighbor of $u$($v$). Let $\kappa$ be a fixed parameter that is implementation dependent (cf. Section \ref{sec: distributed}). For the purposes of this algorithm, we assume the existence of a $\kappa$-regular expander with edge expansion $\alpha > 2$.

At any time step, the adversary can add a node (with its incident edges) or delete a node (with its incident edges). Addition is straightforward, the algorithm takes no action. The added edges are colored black.

The self-healing algorithm is mainly concerned with what edges to add and/or delete when
a node is deleted. The algorithm adds/deletes edges based on the colors of the edges deleted as well as on other factors as described below. Let $v$ be the deleted node and $NBR(v)$ be the neighbors of $v$ in the network after the current deletion.  We have the following cases:

\noindent {\bf Case 1:} {\em All the deleted edges are black edges.} In this case, we construct a $\kappa$-regular expander among the neighbor nodes $NBR(v)$ of the deleted node.  (If the number of neighbors is less than $\kappa$, then a clique (a complete graph) is constructed among these nodes.)
All the edges of this expander are colored by a unique color,  say $C_v$ (e.g., the $\ID$ of the deleted node can be chosen as the color, assuming that every node gets a unique ID whenever it is inserted to the network). Note that the addition of the expander edges is such
that multi-edges are not created. In other words, if (black) edge $(u,v)$ is already present, and the expander construction mandates the addition
of a (colored) edge between $(u,v)$ then this done by simply re-coloring the edge to color $C_v$. Thus our algorithm does not add multi-edges.

We call the expander subgraph constructed in this case among the nodes in $NBR(v)$ as a  {\em primary (expander) cloud} or simply {\em a primary cloud} and all the (colored) edges in the cloud are  called
primary edges. (The term ``cloud" is used to  capture the fact that the nodes involved are ``closeby", i.e., local to each other.) To identify the primary cloud (as opposed to a secondary, described later) we assume that all primary colors are different shades of color {\bf red}.

\noindent {\bf Case 2:} {\em At least some of the deleted edges are colored edges}. In this case, we have two subcases. 

\noindent {\bf Case 2.1:} {\em All the deleted colored edges are primary edges.}
Let the colored edges belong to the colors $C_1, C_2, \dots, C_j$.   This means that the deleted node $v$ belonged to $j$ primary clouds (see Figure \ref{fig: nodeinmanyprimary}). There will be $\kappa$ edges of each color class deleted, since $v$ would have degree $\kappa$ in each of the primary expander clouds. In case $v$ has black neighbors, then some black edges will also be deleted. Assume for sake of simplicity that there are no black
neighbors for now.  If they are present, they can be handled in the same manner
as described later.

In this subcase, we do two operations.
First, we fix each of the $j$ primary clouds. Each of these clouds lost a node and so the cloud is no longer a $\kappa$-regular expander. We reconstruct a new $\kappa$-regular expander in each of the primary clouds (among the remaining nodes of each cloud). (This reconstruction is done in an incremental fashion
for efficiency reasons --- cf. Section  \ref{sec: distributed}.) The color of the edges of the respective primary clouds are retained.
Second, we pick one {\em free} node, if available (free nodes are explained below), from each primary cloud (i.e., there will be $j$ such nodes picked, one from each primary cloud) and these nodes will be connected together via a (new) $\kappa$-regular expander. (Again if the
number of primary clouds involved are less than or equal $\kappa+1$ i.e., $j\leq \kappa+1$, then a clique will be constructed.) 
The edges of this expander will have a new (unique) color of its own. We call the expander subgraph constructed in this case among the $j$ nodes  as a  {\em secondary (expander) cloud} or simply a {\em secondary cloud} and all the (colored) edges in the cloud are  called
{\em secondary} edges. To identify a secondary cloud, we assume that all secondary colors are different  shades of color {\bf orange}. 

If the deleted node $v$ has black neighbors, then they are treated similarly,
consider each of the neighbors as a singleton primary cloud and then proceed
as above.

\begin{figure}[h!]
\label{fig: nodeinmanyprimary}
\centering
\includegraphics[scale=0.8]{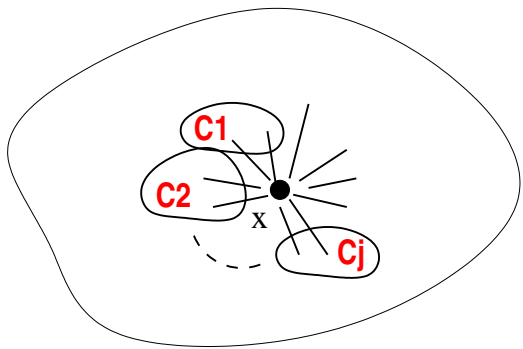}
\caption{A node can be part of many primary clouds.} 
\end{figure}

\noindent {\bf Free nodes and their choosing:}
 The nodes of the primary clouds picked  to form the secondary cloud are called {\em non-free} nodes.
Thus free nodes are nodes that belong to only primary clouds. 
We note that a free node can belong to more than one primary cloud (see e.g., Figure~\ref{fig: nodeinmanyprimary}). In the  above construction of  the secondary cloud, we choose one
unique free node from each cloud, i.e., if there are $j$ clouds then we choose
$j$ different nodes and associate each with one unique primary cloud (if a free node belongs to two or more primary clouds, we associate it with only one of them) such
that each primary cloud has exactly one free node associated with it. (How this is implemented is deferred to Section \ref{sec: distributed}.) 
We call the free node associated with a particular primary cloud as the {\em bridge node} that  ``connects" the primary cloud with the secondary cloud. Note that our construction implies that any (bridge) node of a primary cloud can
belong to at most one secondary cloud.

What if there are no free nodes associated with a primary cloud, say $C$?
Then we pick a free node (say $w$) from another cloud among the $j$ primary clouds (say $C'$) and {\em share} the node with the cloud $C$.
Sharing means adding $w$ to $C$ and forming a new $\kappa$-regular expander among the remaining nodes of  $C$ (including $w$).
Thus $w$ will be part of both $C$ and $C'$ clouds. $w$ will be used as a free node associated with $C$ for the subsequent repair. Note that this 
might render $C'$ devoid of free nodes. To compensate for this, $C'$ gets
a free node  (if available) from some other cloud (among the $j$ primary clouds). Thus, in effect, every cloud will have its own free node associated with it, if there are at least $j$ free nodes (totally) among the $j$ clouds. %(Again, we defer how
%the free nodes get associated with different clouds to Section \ref{sec: distributed}). 

There is only one more possibility left to the discussed. If there are less
than $j$ free nodes among all the $j$ clouds, then we {\em combine} all the $j$ primary clouds into a {\em single} primary cloud, i.e., we construct a $\kappa$-regular expander among all
the nodes of the $j$ primary cloud (the previous edges belonging to the clouds are 
deleted). The edges of the new cloud will have a new (unique) color associated with it. Also  all non-free nodes  associated with the previous $j$ clouds
become free again in the  combined cloud. We note that combining many primary clouds into one primary cloud is
a costly operation (involves a lot of restructuring). We amortize this costly operation over many cheaper operations. This is the  main intuition behind constructing a secondary expander and free nodes; constructing a secondary expander is cheaper than combining many primary expanders and this is not possible only if there are no free nodes (which happens only once in a while).

\noindent {\bf Case 2.2:} {\em Some of the deleted edges are secondary edges.}
In other words, the deleted node, say $v$, will be a  bridge (non-free) node.
Let the deleted edges belong to the primary clouds $C_1, C_2, \dots, C_j$ and
the secondary cloud $F$.    (Our algorithm guarantees that a
bridge node can belong to at most  one secondary cloud.)
We handle this deletion as follows. Let $v$ be the bridge node
associated with the primary cloud $C_i$ (one among the $j$ clouds). Without loss of generality, let the secondary cloud connect a strict subset,
i.e., $j' < j$ primary clouds with possibly other (unaffected) primary clouds.  This case is shown in Figure~\ref{fig: case22sec}. As done in Case 2.1, we first fix all the $j$ primary clouds by constructing a new $\kappa$-regular expander among the remaining nodes. 
We then fix the secondary cloud by  finding another free node, say $z$, from $C_i$, and reconstructing  a new $\kappa$-regular secondary cloud expander on $z$ and other bridge nodes of  other primary clouds of $F$. The edges retain
the same color as their original. If there are no free nodes among all the primary clouds of $F$, then  all
primary clouds  of $F$ are combined into one new primary cloud as explained in Case 2.1 above (edges of $F$ are deleted). 
The remaining $j-j'$ primary clouds are then repaired as in case 2.1 by constructing a secondary cloud between them.

\begin{figure}[h!]
\label{fig: case22sec}
\centering
\includegraphics[scale=0.8]{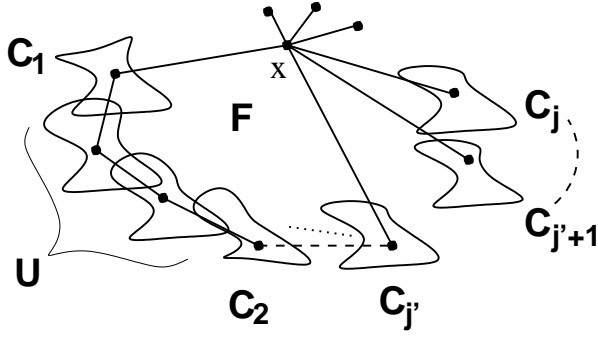}
\caption{Case 2.2: Deleted node x part of secondary cloud F, and primary clouds} 
\end{figure}

%\begin{}[h!]

\begin{algorithm}[h!]
\label{algo: repairbyexpander}
\caption{\textsc{Xheal}($\G, \kappa$)}
%\begin{boxedminipage}{0.5\textwidth}
\begin{algorithmic}[1]

 \IF{node $v$ inserted with  incident edges}
\STATE The inserted edges are colored black.
\ENDIF

\IF{node $v$ is deleted}
\IF{all deleted edges are black} 
%\COMMENT{Case 1}
\STATE   \textsc{MakeCloud}($BlackNbrs(v), primary, Clr_{new}$)
\ELSIF{deleted colored edges are all primary}
%[Case 2.1]
\STATE Let $C_1, \dots, C_j$ be primary clouds that lost an edge
\STATE  \textsc{FixPrimary}($[C_{1}, \dots, C_{j}]$)
\STATE \textsc{MakeSecondary}($[C_1, \dots, C_j] \cup BlackNbrs(v)$)
\ELSE
% \Comment{Case 2.2}
\STATE Let $[C_1, \dots, C_j] \leftarrow$ primary clouds of $v$; $F \leftarrow$ secondary cloud of $v$; $[U] \leftarrow$  $Clouds(F) \setminus [C_1, \dots, C_{j}]$, $[C_1, \dots, C_{j'}] \leftarrow F \cap [C_1, \dots, C_{j}]  $  
\STATE  \textsc{FixPrimary}($[C_{1}, \dots, C_{j}]$)
\STATE   \textsc{FixSecondary}($F, v$)
\STATE \textsc{MakeSecondary}($[C_{j' + 1}, \dots, C_j] \cup BlackNbrs(v)$)
\ENDIF
 \ENDIF
\end{algorithmic}
\end{algorithm}

\begin{algorithm}[h!]
\label{algo: makecloud}
\caption{\textsc{MakeCloud}($[V], Type, Clr$)}
\begin{algorithmic}[1]
\IF {$|V| \leq \kappa+1$}
\STATE Make clique among $[V]$
\ELSE
\STATE Make $\kappa$-reg expander among $[V]$ of edge $(Type, Clr)$
\ENDIF
\end{algorithmic}
\end{algorithm}

\begin{algorithm}[h!]
\label{algo: fixprimary}
\caption{\textsc{FixPrimary}($[C]$)}
%\begin{boxedminipage}{0.5\textwidth}
\begin{algorithmic}[1]
\FOR{each cloud $C_i \in [C]$ }
\STATE  \textsc{MakeCloud}($C_i, primary, Color(C_i)$)   
\ENDFOR
\end{algorithmic}
%\end{boxedminipage}
%\end{figure}
\end{algorithm}

\begin{algorithm}[h!]
\label{algo: makesecondary}
\caption{\textsc{MakeSecondary}($[C]$)}
%\begin{boxedminipage}{0.5\textwidth}
\begin{algorithmic}[1]
\FOR{each cloud $C_i \in [C]$ }
\IF{$FrNode_i$ = \textsc{PickFreeNode}($C_i$) == NULL} 
\STATE  \textsc{MakeCloud}($Nodes([C]), primary, Clr_{new}$)   
% \COMMENT{Free node not available, reconstruct}
\STATE Return
\ENDIF
\ENDFOR
\STATE  \textsc{MakeCloud}($\bigcup FrNode_i\ \forall C_i \in [C]$, secondary,$Clr_{new}$)
\end{algorithmic}
%\end{boxedminipage}
%\end{figure}
\end{algorithm}

\begin{algorithm}[h!]
\label{algo: fixsecondary}
\caption{\textsc{FixSecondaryCloud}($F, v$)}
%\begin{boxedminipage}{0.5\textwidth}
\begin{algorithmic}[1]
\IF{$v$ is a bridge node of $C_i$ in $F$} 
\IF{$FrNode_i$ = \textsc{PickFreeNode}($C_i$) == NULL} 
\STATE  \textsc{MakeCloud}($Nodes(F), primary, Clr_{new}$)   
% \COMMENT{Free node not available, reconstruct}
\ELSE
\STATE  \textsc{MakeCloud}($FrNode_i \cup BridgeNode(C_j)\ \forall C_i \in [C]$, secondary, Color(F))    
\ENDIF
\ENDIF
\end{algorithmic}
%\end{boxedminipage}
%\end{figure}
\end{algorithm}

\begin{algorithm}[h!]
\label{algo: pickfreenode}
\caption{\textsc{PickFreeNode}()}
%\begin{boxedminipage}{0.5\textwidth}
\begin{algorithmic}[1]

\STATE Let a Free node be a primary node without secondary duties
\IF{Free node in my cloud}
\STATE Return Free node
\ELSE
\STATE Ask neighbor clouds; if a free node found, return node, else return NULL
\ENDIF
\end{algorithmic}
%\end{boxedminipage}
%\end{figure}
\end{algorithm}

\section{Analysis of  \textbf{\it Xheal}}
\label{sec: Results}

The following is our main theorem on the guarantees that \Xh provides on the topological properties of the healed graph. The theorem assumes that \Xh is able to construct a $\kappa$-regular expander (deterministically) of expansion $\alpha > 2$.

\begin{theorem}
\label{th:main}
 For graph $G_t$(present graph) and graph $G'_t$ (of only original and inserted edges), at any time t, where a timestep is an insertion or deletion followed by healing:
 \begin{enumerate}
  \item For all $x \in G_t$, $degree_{G_t}(x) \le \kappa.degree_{G'_t}(y)$, for a fixed constant $\kappa> 0$.
  \item For any two nodes $u,v \in G_t$,  $\delta_{G_t}(u,v) \le \delta_{G'_t}(u,v) O(\log n)$, where $\delta(u,v)$ is the shortest path between $u$ and $v$,
  and $n$ is the number of nodes in $G_t$.
   \item $h(G_t) \ge min(\alpha, h(G'_t))$, for some fixed constant  $\alpha \ge 1$.
   \item $\lambda(G_t) \ge min\left(\Omega\left(\frac{\lambda(G'_t)^2 d_{min}(G'_t)}{(\kappa)^2(d_{max}(G'_t))^2}\right), \Omega\left(\frac{1}{(\kappa d_{max}(G'_t))^2}\right)\right)$, where $d_{min}(G'_t)$ and $d_{max}(G'_t)$ are the minimum
   and maximum degrees of $G'_t$.
    \end{enumerate}
\end{theorem}

From the above theorem, we get an important corollary:

\begin{corr}
\label{cor:spectral}
If $G'_t$ is a (bounded degree) expander, then so is $G_t$.  In other words,
if the original graph and the inserted edges is an expander, then \Xh guarantees that the healed graph also is an expander.
\end{corr}

%\pagebreak 

\subsection{Expansion, Degree and Stretch}
\label{subsec: upperbounds}

\begin{lemma}
\label{lemma: firstheal}
Suppose at the first timestep (t=1), a deletion occurs. Then, after healing, $h(G_1) \ge min(c,h(G'_1))$, for a constant $c  \geq 1$.
\end{lemma}

\begin{proof}
\begin{figure}[h!]
\label{fig: expandercuts}
\centering
\includegraphics[scale=0.6]{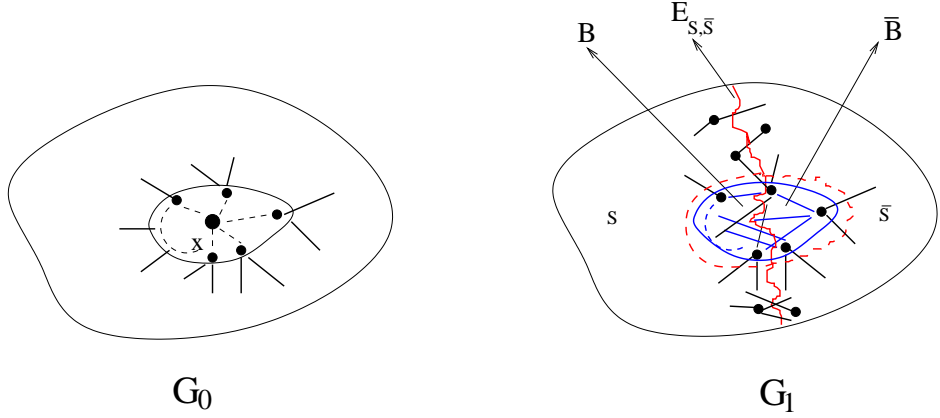}
\caption{Healed graph after deletion of node $x$. The ball of $x$ and its neighbors gets replaced by a $\kappa$-regular expander of its neighbors --- Case 1 of the Algorithm.} 
\end{figure}

%For ease of notation, refer to the original graph before the first adversarial deletion as  $G$ and the healed graph after the first deletion as $H$ (Figure~\ref{fig: expandercuts}) i.e. $G = G'_1, H = G_1$.  Suppose the node deleted at t=1 is the node $x$. Consider the ball formed by $x$ and its neighbors. We can consider the healing to have replaced this ball by another one which is a k-expander of the ex-neighbors of $x$. Let us color blue the ex-neighbors of $x$, and all the edges between these neighbors (including the new ones added by the algorithm). Consider the rest of the edges as black.
%
% Consider all the cuts which can define expansion in $H$ i.e.cuts and a  component of size not more than $n/2$, where $n$ is the total number of nodes in the graph. We are interested in a cut that yields the minimum expansion. Let the set of all such cuts in $H$ be $M(S)$, one such cut be $\E_S(H) \in M(S)$  and the corresponding connected component be $S(H)$.
% % and $\E_S(G)$ in $G$  and $S(G)$ respectively.  

Observe that the initial graphs $G_0$ and $G'_0$ are identical. Suppose that node $x$ is deleted at $t=1$.
For ease of notation, refer to the graph $G_0$ as $G$ and the healed graph $G_1$ as $H$. Notice that $G'_1$ is the same as $G_0$, since the graph $G'_t$ does not change if the action at time $t$ is a deletion. Consider the induced subgraph formed by $x$ and its neighbors. Since all the deleted edges are black edges,
Case 1 of the algorithm applies. Thus the healing algorithm will replace this subgraph by  a new subgraph, a $\kappa$-regular expander over $x$'s ex-neighbors. Let us call this new subgraph $I$.  Note that this corresponds to {\bf Case 1} of the Algorithm. We refer to Figure \ref{fig: expandercuts}.
 
 Consider a set $S(H)$ which defines the expansion in $H$ i.e. $|S(H)| \le n/2$ (where $n$ is the number of nodes in $G$), and $S(H)$ has the minimum expansion over all the subsets of $H$. Call the cut induced by $S(H)$ as $E_{S,\bar{S}}(H)$ and its size as $|E|_{S,\bar{S}}(H)$. Also refer to the same set in $G$ (without $x$ if $S(H)$ included $x$) as $S(G)$, and the cut as $E_{S,\bar{S}}(G)$.  The key idea of the proof
 is to directly bound the expansion of $H$, instead of looking at the change of expansion from
 of $G$.
 % and computing how much has the expansion has changed because of the deletion; the former turns out to be easier to bound.
  In particular,
 we have to handle the possibility that our self-healing algorithm
 may not add any new edges, because those edges may already be present.
 (Intuitively, this means that the prior expansion itself is  good.)
  
 We consider two cases depending on  whether the healing may or may not have affected this cut.  
% We have the following cases:

\begin{enumerate}
%\item{\textbf{Any $\E_S(H) \in M(S)$ has only black edges:}}
\item{$E_{S,\bar{S}}(H) \cap E(I) = \emptyset$:}
\label{prfpart: onlyblackedges}

This implies that only the edges which were in $G$ are involved in the cut $E_{S,\bar{S}}(H)$.  
 Since expansion is defined as the minimum over all cuts,  $|\E|_{S,\bar{S}}(G) \ge h(G) |S(G)|$.
 Also, since $\E_{S,\bar{S}}(H) = \E_{S,\bar{S}}(G)$ and $S(H) \le S(G)$, we have: 
 \[
 h(H) = \frac{\E_{S,\bar{S}}(H)}{S(H)} \ge \frac{\E_{S,\bar{S}}(G)}{S(G)} \ge h(G).
 \]
 
% \item \textbf{Every $\E_S(H) \in M(S)$  has some blue edges:} 
 \item{$E_{S,\bar{S}}(H) \cap E(I) \ne \emptyset$:}
  Notice that if there is any minimum expansion cut not intersecting $E(I)$,  part~\ref{prfpart: onlyblackedges} applies, and we are done.
  
  %Thus, all the minimal cuts must pass through the `blue ball'.  
  
%  Let $\E_A$ be the part of the cut that passes through the ball, $\E_{\bar{A}}$ the part outside the ball, and let $B$ be the blue nodes that are part of $S(H)$.
% Now consider the corresponding cut $\E(G)$ in $G$. This cut has the same edges except for those lost to the node $x$. Thus, the contribution of the edges of $G$ to $\E_S(H)$ is more than $ |S(H)|.h(G) - |B|$. Also, the healing algorithm tries to add $k$ new blue edges per node of $B$. Note that it may not succeed (e.g. when the graph is a clique) but it will add enough blue edges to make the blue ball at least a $k-expander$, if it were not already so. There are the following cases:
 
 The healing algorithm tries to add enough new edges (if needed) into $I$ so that $I$ itself has an expansion of $\alpha > 2$ (cf. Algorithm in Section \ref{sec:algorithm}). Note that it may not succeed if $|I|$ is too small. However, in that case, the algorithm  makes $I$ a clique and achieves an expansion of  $c$ where $c \geq 1$. Thus, we have the following cases:

 \begin{enumerate}
 % Amitabh, Dec 28: Case removed (not needed seperately).
 
% \item \emph{The algorithm adds $k$ new edges per node of $B$:}\\
%Intuitively, the algorithm would be able to do so if $h(G)$ was low and there were enough neighbors of $x$. In this case:
%\[
% h(H) = \frac{\E_S(H)}{S(H)} \ge \frac{|S(H)|.h(G) - |B| + |B|.k}{|S(H)|} = h(G) + \frac{|B|}{|S(H)|}(k-1) > h(G)
% \]

% \item \emph{The algorithm adds less than $k$ new edges per node of $B$:}\\
 \item \emph{$I$ has an expansion of  $\alpha > 2$:}\\
%Intuitively, this may happen if h(G) was already high, the graph was dense enough, or $x$ did not have  enough neighbors. 
 Consider the nodes in $I$ which are part of $S(H)$ i.e., $B  = S(H) \cap I$. We want to calculate $h(H)$. Since expansion is defined over sets of size not more than half of the size of the graph, we can do so in two ways:

%Thus, there are two cases:
\begin{enumerate}
%\item \emph{$x$ had more than $k$ neighbors:}
\item {$B \le I/2$:}
$S(H)$ expands at least as much as $h(G)$ except for the edges lost to $x$, and our algorithm ensures that $I$ has expansion of at least $\alpha > 2$. Therefore, we have:
%the blue ball is at least a k-expander. Let us say the algorithm borrows some blue edges from the existing ones to complete the expander. Thus, we  have:
\begin{eqnarray*}
 h(H) & = &\frac{\E_{S,\bar{S}}(H)}{S(H)}\\
 &   \ge&  \frac{(|S(H)| - |B|).h(G) - |B| + |B|. \alpha}{|S(H)|}\\
  &  =  & \frac{(|S(H)| - |B|).h(G) +  |B|.(\alpha -1)}{|S(H)|}
      \end{eqnarray*}
In the numerator above, we have $(|S(H)| - |B|).h(G)$ which is a lower
bound for the number of edges emanating from the set $S(H)$ (we minus $|B|$ from 
$|S(H)|$
to account for the edges that may be already present, note that Xheal does 
not add edges between two nodes if they are already present.) We subtract another $|B|$ or the edges lost to the deleted node and add $|B| \alpha$ edges due to the expansion gained.

The following cases arise:
 If $h(G) \ge  \alpha - 1$, we have  $h(H)  \ge   \frac{|S(H)|(\alpha-1) }{|S(H)|}
   \ge  \alpha - 1 > 1$. Otherwise, 
%\begin{eqnarray*}
% h(H) & \ge &  \frac{(|S(H)| - |B|)(\alpha-1)  +  |B|.(\alpha -1)}{|S(H)|}\\
 % & \ge & \alpha - 1 > 1
%\end{eqnarray*}
if $  h(G) \le \alpha - 1$,  we get:  $ h(H) \ge   \frac{|S(H)|.h(G)}{|S(H)|}
   \ge  h(G) $

%\begin{eqnarray*}
% h(H) & \ge &  \frac{(|S(H)| - |B|).h(G)  +  |B|.h(G)}{|S(H)|}\\
%  & \ge & h(G)
%\end{eqnarray*}

\item {$\bar{B} \le I/2$:}
By construction, nodes of $\bar{B}$ expand with expansion at least $\alpha$ in the subgraph $I$.
 %Consider the contribution of $\bar{B}$ to the cut $ \E_{S,\bar{S}}(H)$, and calculate $h(H)$. 
  Similar to above, we get, $ h(H) \ge  \frac{(|S(H)| - |\bar{B}|).h(G) +  |\bar{B}|.(\alpha -1)}{|S(H)|} $.
%\begin{eqnarray*}
% h(H) & = &\frac{\E_{S,\bar{S}}(H)}{S(H)}\\
% &   \ge&  \frac{(|S(H) \cup {x}| - |\bar{B}| ).h(G) - |\bar{B}| + |\bar{B}|. \alpha}{|S(H)|}\\
% & \ge & \frac{(|S(H)| - |\bar{B}|).h(G) - |\bar{B}| + |\bar{B}|. \alpha}{|S(H)|}\\  
% &  =  & \frac{(|S(H)| - |\bar{B}|).h(G) +  |\bar{B}|.(\alpha -1)}{|S(H)|} 
% \end{eqnarray*}
%
Thus, if $h(G) \ge  \alpha - 1$, then $ h(H) \ge \alpha - 1$, else  $ h(H)  \ge h(G)$.

 \end{enumerate}

%\item \emph{$x$ did not have more than $k$ neighbors:} \\
 \item \emph{$I$ has an expansion of  $c < \alpha$:}\\
 This  happens in the case of the degree of $x$ being smaller than $k$.  In this case,  the expander  $I$  is just a clique. Note that, even if degree of $x$ is 2, the expansion is 1. (When the degree of $x$ is 1, then the deleted node is
 just dropped, and it is easy to show that in this case, $h(H) \geq h(G)$.) The same analysis as the above applies, and we get $h(H) \ge min(c', h(G))$, for some constant $c'\geq 1$. Since $G$ is $G_1$ and $H$ is $G'_1$, we get  $h(G_1) \ge min(c', h(G'_1))$.
%\qedhere
 % In the special case of $x$ having just one neighbor, no healing is required.
% In this case, we cannot construct a k-expander. However, we can construct a c-expander  i.e. a c+1-clique, where $c < k$, and the above analysis applies. Since $k \ge 2$, unless $x$ was a corner/`leaf' node, in which case no healing is required, $c  \ge 1$.
\end{enumerate}
\end{enumerate}
%\qedhere
\end{proof}

\begin{corr}
Given a graph $G$, and a subgraph $B$ of $G$, construct a new graph $H$ as follows: Delete the edges of $B$ and insert an expander of expansion $\alpha > 2$  among the nodes of $B$. Then $h(H) \ge \min(c,h(G))$, where $c$ is a constant.
\end{corr}

%The following lemmas have proofs defered to the complete version.

\begin{lemma}
\label{lemma: Expansion}
At  end of any timestep t, $h(G_t) \ge min(c',h(G'_t))$, where $c' \geq 1$ 
is a fixed constant.
\end{lemma}

%\begin{lemma}
%\label{lemma: ExpansionApp}
%At the end of any timestep t, $h(G_t) \ge min(c',h(G'_t))$, where $c' \geq 1$ 
%is a fixed constant.
%\end{lemma}
\begin{proof}
%Expand proof of lemma~\ref{lemma: firstheal}. Use lemma~\ref{lemma: IntraExpander}, and intra-cloud expanders.

\begin{figure}[h!]
\centering
\includegraphics[scale=0.6]{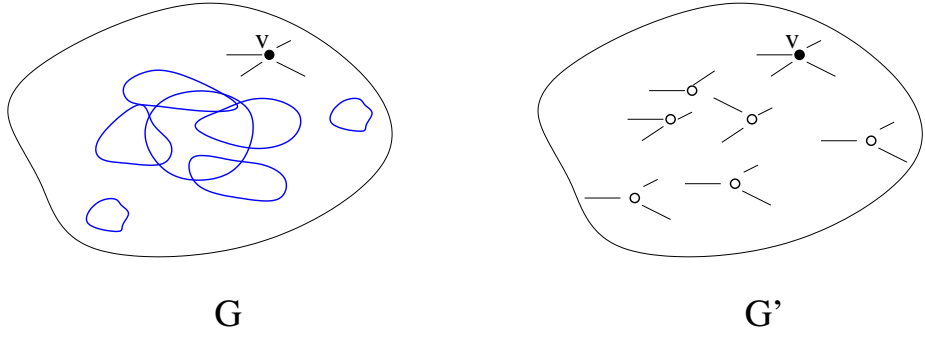}
\label{fig: insertcuts}
\caption{Graphs $G_t$ and $G'_t$ after insertion of node $v$. Graph $G_t$ has some colored clouds. The nodes which have already been deleted are present in graph $G'_t$ and are shown as unfilled nodes.} 
\end{figure}

\begin{figure}[h!]
\centering
\includegraphics[scale=0.6]{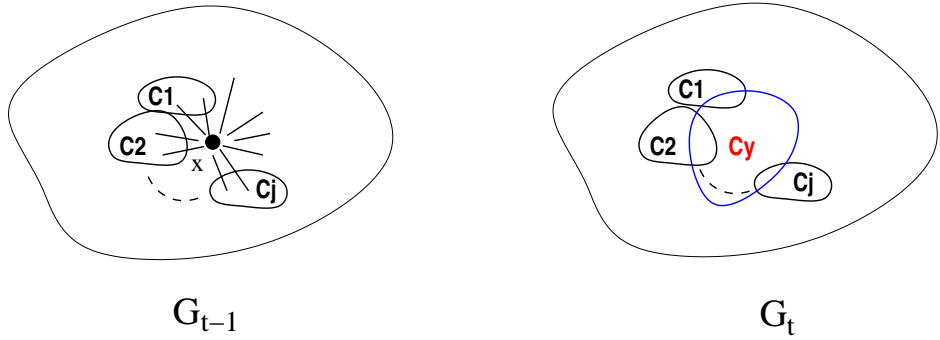}
\label{fig: manyexpandercuts}
\caption{Healed graph after deletion of node $x$. The 'black' neighbors of  $x$  and some neighbors of $x$ from color clouds $C_1, C_2, \ldots C_j$ get connected by a $\kappa$-regular  expander of color $C_y$.} 
\end{figure}

First, consider the case when node $v$ is inserted at time $t$. Observe that the topologies of both the graphs $G_t$ and $G'_t$ would be the same if all the insertions were to happen before the deletions. This is because an incoming node comes in with only black edges and at no step does the healing algorithm rely on the  number of nodes present or uses edges for possible future nodes. Therefore, for our analysis, consider an order in which all the insertions happened before the first deletion, in particular think of  node $v$ as being inserted at time $s$, and the first deletion happening at time $s+1$. Since the  graphs $G_i$ and $G'_i$ would look exactly the same for all $i$ before $s+1$, insertion of node $v$ changes both the graphs $G_t$ and $G'_t$ in exactly the same way. Thus, if we can show that our lemma holds when a deletion happens (as we show below), we are done.

%Now, consider the case when node $v$ is inserted at time $t$.
%%% Fix this
% Both the graphs $G_t$ and $G'_t$ get updated in the same way. This also means that  $G'_t$ has exactly the same black edges for the black nodes in $G_t$ (i.e. the nodes who have lost no neighbor).  Since the number of nodes in $G_t$ is less than or equal to that in $G'_t$, if the expansion cut in $G'_t$ is across only black edges,  the corresponding cut in $G_t$ is at least as good, and we are done.  If the expansion cut involves colored edges, then, the analysis as above holds, using $G_t$ as $H$ and $G'_t$ as $G$, and we get  $h(G_t) \ge min( c',  h(G'_t))$, for  some $c' = min(c, \alpha - 1)$.
%%%

Next, we consider that a deletion occurs at timestep $t$.  
The proof will be by induction on $t$. Lemma~\ref{lemma: firstheal} already shows the base case, where it is assumed \emph{wlog} that the first deletion occurs at time $t=1$. Notice that before the first deletion, graphs $G$ and $G'$ are identical and the proof is trivial. 

As per the algorithm, we have two main cases to consider.

{\bf Case 1:} This case occurs when the deleted edges are all black edges. This case is handled exactly as in the proof of  Lemma \ref{lemma: firstheal}.
 
{\bf Case 2.1 and Case 2.2:} We analyze Case 2.1 below, the analysis of Case 2.2 is similar.

 First, we give the proof assuming that each cloud has a free node associated with it.
 
Refer to figure~\ref{fig: manyexpandercuts}. Let $G$ be the original graph and $H$ the healed graph.  Let $x$ be the node deleted. The graph $G$ corresponds to the graph $G_{t-1}$. The graph $G'_{t-1}$ is the same as the graph $G'_t$ since the graph $G'$ does not change on deletion.
By the induction hypothesis, $h(G) = h(G_{t-1}) \ge min(c, h(G'_{t-1}) =  h(G'_{t}) )$. The graph $H$ corresponds to the healed graph $G_t$. Thus, if we show $h(H) \ge  h(G)$, we are done.

In this case, let the deleted node $x$ belong to $j$ primary clouds $C_1$ to $C_j$. (We note that if $x$ has black neighbors, the algorithm treats them as singleton primary clouds.) First the primary clouds are restructured by constructing a new $k$-regular expander among the remaining nodes of the cloud (excluding the deleted node). Then, a free  node from each color cloud is picked and are connected to form a $\kappa$-regular expander of color, say, $C_x$ --- this is the secondary cloud.

The proof is a generalization of the argument  of Lemma~\ref{lemma: firstheal}. Let $E_S((H)$ be a cut that defines the expansion in the graph $H$, and $S_H$ as defined before . Let us call this a minimum cut. If any minimum cut $E_S(H)$ passes through only the edges
of $E(G)-(E(C_1)\cup \dots \cup E(C_j) \cup E(C_x)$ (i.e., outside these clouds) then the  expansion of $H$ cannot decrease and we are done. Thus, we will consider the cases when all minimum cuts pass through some edges of the above clouds.

%Let the neighbors of the deleted node $x$ in cloud $C_i$ 
%be $F_i$ ($1 \leq i \leq j$). Note that $|F_i| = k$ neighbors ($1 \leq i \leq j$). 
 Each of the colored balls maintains an expansion of at least $\alpha > 2$. Let $B_1, B_2, \ldots B_j$, $B_x$, be the nodes of $S_H$ in the balls of color $C_1, C_2,  \ldots C_j$, $C_x$ respectively. (We  abuse notation so that each $C_i$ also denotes the subgraph defining the respective primary cloud.) In the following, for $1 \leq i \leq j$,  we define $A_i = B_i$ if $|B_i| \leq |C_i|/2$, otherwise,  we define $A_i = \bar{B_i} = C_i - B_i$ if $|B_i| \geq |C_i|/2$.  
 $A_x$ is similarly defined. 
 %Note that $ |H| < |G|$, thus, $S_H \le S_G$.
   
We have:
 
% \begin{eqnarray*}
% h(H) & \ge &  \frac{(|S(H)| - |A_1| - |A_2| \ldots - |A_x|) h(G)  - (|A_1| + |A_2| \ldots + |A_x|) + (|A_1| + |A_2| \ldots + |A_x|) \alpha}{|S(H)|} \\
%  & = &  \frac{(|S(H)| - |A_1| - |A_2| \ldots - |A_x|) h(G)  + (|A_1| + |A_2| \ldots + |A_x|)(\alpha - 1)}{|S(H)|}
%\end{eqnarray*}

\begin{eqnarray*}
 h(H) & \ge &  \frac{(|S(H)| -\sum_{i= 1}^x |A_i|) h(G)  - \sum_{i= 1}^x |A_i| +  \sum_{i= 1}^x |A_i| \alpha}{|S(H)|} \\
  & = &  \frac{(|S(H)| - \sum_{i= 1}^x |A_i|) h(G)  + (\sum_{i= 1}^x |A_i|)(\alpha - 1)}{|S(H)|}
\end{eqnarray*}

 If $h(G) \ge  \alpha - 1$, we have:
 
\begin{eqnarray*}
 h(H) & \ge  \alpha - 1 > 1
\end{eqnarray*}

 If $  h(G) \le \alpha - 1$,  we have:
 
\begin{eqnarray*}
 h(H) & \ge  & h(G) 
\end{eqnarray*}

Thus, $h(G_t) = min(c',  h(G'_t))$, for  some $c' = min(c, \alpha - 1)$ and
the induction hypothesis holds.

The above analysis assumes that each primary cloud had a free node for itself.
Otherwise, as per the algorithm, free nodes from other clouds are shared. If there there a total of $j$ free nodes among all the $j$ clouds,
then also the analysis  proceeds as above. The only difference is that when a free node is shared
between two clouds, its degree increases (by $k$). This can only increase the expansion, and hence the above analysis goes through.
The other possibility is that there are less than $j$ free nodes. In this case,
all the primary clouds are combined into one single expander cloud.
Here also, the analysis is similar to above.
\end{proof}

%\begin{proof}
% The proof uses Lemma~\ref{lemma: firstheal}, and is given in the appendix.
%\end{proof}

%\subsection{Degree and Stretch Analysis}

%\subsection{Degree Analysis}
\begin{lemma}
\label{lemma: degree}
  For all $x \in G_t$, $degree_{G_t}(x) \le O(\kappa.degree_{G'_t}(x))$,  for a fixed parameter $\kappa> 0$.
\end{lemma}

%\begin{lemma}
%\label{lemma: degreeApp}
%  For all $x \in G_t$, $degree_{G_t}(x) = O(\kappa.degree_{G'_t}(x))$,  for a fixed parameter $\kappa > 0$.
%\end{lemma}

\begin{proof}
We bound the increase in degree of any node $x$ that belongs to both
$G_t$ and $G'_t$. Let the degree of $x$ in $G'_t$ be $d'(x) = degree_{G'_t}(x)$.
This will be black-degree of $x$ (as $G'_t$ comprises solely of edges present
in the original graph plus the inserted edges). There are three cases to consider and we bound the degree increase
in each:

1. Whenever, a black edge gets deleted from this node, the self-healing algorithm, adds $\kappa$ colored edges in place of it, because
a $\kappa$-regular expander is constructed which includes this node (this
expander can be a primary or a secondary cloud). Thus $x$'s degree can increase
by a factor of $\kappa$ at most because of deletion of black edges.

2.  When $x$ loses
a colored edge, then the algorithm restructures the expander cloud by constructing a new $\kappa$-regular expander. Again, this is true if the reconstruction is done on a primary or a secondary cloud. In this case,
the degree of $x$ does not change.  

3. Finally, we consider the effect of non-free nodes. $x$'s degree can increase if it is chosen as a bridge (non-free) node to connect
a primary cloud (with which it is associated) to a secondary cloud. In this case, its degree will increase by $\kappa$, since it will become part of the secondary cloud expander. There is one more possibility
that can contribute to increase of $x$'s degree by $\kappa$ more. If $x$ is chosen to be shared as a free node, i.e., it gets associated as a free node with another primary cloud than it originally belongs to, then its degree increases by $\kappa$ more, since it becomes part of another $\kappa$-regular expander. The shared node
becomes a bridge node, i.e., a non-free node in that time step. Hence it cannot
be shared henceforth.

From the above, we can bound the degree of $x$ in $G_t$,  $d(x) = degree_{G_t}(x)$, as follows: 
$d(x) \leq \kappa  d'(x) + 2\kappa.$
The lemma follows.
\end{proof}

%\begin{proof}
%The proof is deferred to the appendix.
%\end{proof}

\begin{lemma}
\label{lemma: stretch}
   For any two nodes $u,v \in G_t$, \\ $\delta_{G_t}(u,v) \le \delta_{G'_t}(u,v).O(\log n)$, where $\delta(u,v)$ is the shortest path between $u$ and $v$,  and $n$ is the total number of nodes in $G_t$.
\end{lemma}

%\begin{lemma}
%\label{lemma: stretch}
%   For any two nodes $u,v \in G_t$,  $\delta_{G_t}(u,v) \le \delta_{G'_t}(u,v).O(\log n)$, where $\delta(u,v)$ is the shortest path between $u$ and $v$,
%   and $n$ is the total number of nodes in $G_t$.
%\end{lemma}

\begin{proof}
We fix two nodes $u$ and $v$ and let the shortest distance between
them in $G'_t$ be $\ell$. Since this is on the graph
$G'_t$ (which comprises the original edges plus inserted edges),  all the edges on this path will be black edges. Let this shortest path be denoted by $P = <u, u_1, \dots, u_{\ell-1}, v>$. We assume that $\ell > 1$, because the path will just be the edge $(u,v)$ if $\ell = 1$ in which case there is nothing to prove (the edge will also
be present in $G_t$). 

If all the intermediate nodes are present, then the result follows trivially.
Otherwise, let $u'_1, u'_2, \dots, u'_i$ ($i \leq \ell$) be the $i$ deleted nodes
listed in the order of their deletion (i.e., $u'_1$ was deleted before $u'_2$ and so on).

We show that each node deletion can increase the distance between $u$ and $v$
by  $O(\log n)$. Consider the deletion of node $u'_1$. This will create
a $k$-regular expander (primary or secondary, the latter case will arise if some
incident edges of $u'_1$ are colored) among the neighbors of $u'_1$ in path $P$. Thus the  distance between these neighbors of $u'_1$ will increase by $O(\log (deg(u'_1)) = O(\log n)$. We distinguish two cases for subsequent deletions:

1. When the deleted node, say $u'_{j}$, results in a primary cloud:
In this case, the distance between the neighbors of $u'_j$ will increase
by at most  $O(\log n)$, as above. Note that any subsequent deletion
of nodes belonging to the primary cloud will still keep the same stretch,
as there will always be connected via a $k$-regular expander.

2. When the deleted node,  say $u'_{j}$, results in a secondary cloud:
In this case, there are two possibilities: (a) If the secondary cloud 
does not comprise primary clouds formed from previous deletions of nodes
in the path $P$. In this case, the increase in distance is $O(\log n)$ as above;
(b) If the secondary cloud comprises primary clouds formed from prior deletions of nodes in $P$, then the distance between $u$ and $v$ increases
also by $O(\log n)$,
% \amitabh{Insert figure}
 as one has to traverse through the secondary cloud  (connecting the primary clouds).

Thus, the overall distance between $u$ and $v$ increases by a factor of $O(\log n)$ in $G_t$ compared to the distance in $G'_t$. 

\end{proof}

%\begin{proof}
% The proof is deferred to the appendix.
%
%\end{proof}

\subsection{Spectral Analysis}

%Global network metrics
%--- which depend on the collective behavior of all nodes
%and links ---  characterize global network  performance
%such as routing, congestion, sampling, information propagation, etc.  These in turn depend on global topological properties such as connectivity,  diameter and on the spectral properties of the underlying graph. 
%

We derive bounds on the second smallest eigenvalue $\lambda$ which is closely related
to properties such as mixing time, conductance etc. While it is directly difficult to derive bounds on $\lambda$, we use our bounds on edge expansion 
and the Cheeger's inequality to do so.

We need the following simple inequality which relates the Cheeger constant 
$\phi(G)$ and the edge expansion $h(G)$ of a graph $G$ which follows from their respective definitions. We use $d_{max}(G)$ and $d_{min}(G)$ to denote the maximum and minimum node degrees in $G$.

\begin{eqnarray}
\label{eq:ch-ed}
\frac{h(G)}{d_{max}(G)} \leq \phi(G) \leq \frac{h(G)}{d_{min}(G)}.
\end{eqnarray}

\begin{proof}
By Cheeger's inequality  and by inequality \ref{eq:ch-ed} we have,
$$\lambda(G_t) \geq \frac{\phi(G_t)^2}{2} \geq \frac{1}{2}\left(\frac{h(G_t)}{d_{max}(G_t)}\right)^2$$

By Lemma \ref{lemma: Expansion}, we have,
$h(G_t) \geq min(c', h(G'_t))$, for  some $c' \geq 1$.

So we have two cases:

Case 1: $h(G_t) \geq h(G'_t)$.  
%$\lambda(G_t) \geq \frac{1}{2}\left(\frac{h(G'_t)}{d_{max}(G_t)}\right)^2$.\\
By  using the other half of Cheeger's inequality, and inequality \ref{eq:ch-ed},
and Lemma \ref{lemma: degree} we have:
%$$\lambda(G_t) \geq \frac{1}{2}\left(\frac{h(G'_t)}{d_{max}(G_t)}\right)^2 \geq \frac{1}{2}\left(\frac{\lambda(G'_t)d_{min}(G'_t)}{2d_{max}(G_t)}\right)^2
%\geq \frac{\lambda(G'_t)^2}{8(\kappa)^2} \frac{d_{min}(G'_t)}{(d_{max}(G'_t))^2}
%= \Omega\left(\lambda(G'_t)^2 \frac{d_{min}(G'_t)}{(\kappa)^2(d_{max}(G'_t))^2}\right).$$

\begin{eqnarray*}
\lambda(G_t)  & \geq  & \frac{1}{2}\left(\frac{h(G'_t)}{d_{max}(G_t)}\right)^2 \\
  & \geq & \frac{1}{2}\left(\frac{\lambda(G'_t)d_{min}(G'_t)}{2d_{max}(G_t)}\right)^2 \\
 & \geq & \frac{\lambda(G'_t)^2}{8(\kappa)^2} \frac{d_{min}(G'_t)}{(d_{max}(G'_t))^2} \\
 & = & \Omega\left(\lambda(G'_t)^2 \frac{d_{min}(G'_t)}{(\kappa)^2(d_{max}(G'_t))^2}\right).
\end{eqnarray*}

Case 2: $h(G_t) \geq 1$:

This directly gives:

%$$\lambda(G_t) \geq \frac{1}{2}\left(\frac{1}{d_{max}(G_t)}\right)^2 \geq \Omega\left(\frac{1}{(d_{max}(G_t))^2}\right) \geq \Omega \left(\frac{1}{(\kappa d_{max}(G'_t))^2}\right).$$

\begin{eqnarray*}
\lambda(G_t) & \geq & \frac{1}{2}\left(\frac{1}{d_{max}(G_t)}\right)^2 \\ 
&  \geq & \Omega\left(\frac{1}{(d_{max}(G_t))^2}\right)\\
& \geq & \Omega \left(\frac{1}{(\kappa d_{max}(G'_t))^2}\right).
\end{eqnarray*}
\end{proof}

\iffalse
From the above theorem, we get an important corollary:

\begin{corr}
\label{cor:spectral}
At the end of any timestep t, $\lambda(G_t) = \Omega(\lambda(G'_t)^2)$,
provided, $d_{min}(G'_t) = \Theta(d_{max}(G't))$.
Furthermore, if $G'_t$ is a bounded degree expander, then so is $G_t$.
\end{corr}
\fi

\section{Distributed Implementation of  \textbf{\it Xheal}: Time and Message Complexity Analysis} 
\label{sec: distributed}
We now discuss how to efficiently implement \Xh. 
A key task in \Xh involves the distributed construction and maintenance (under  insertion and deletion) of a regular expander.  We use a randomized construction of Law and Siu \cite{lawsiu} that is described below. The expander graphs of \cite{lawsiu} are formed by constructing a
class of regular graphs  called {\em H-graphs}. An H-graph is
a $2d$-regular multigraph in which the set of edges is composed
of $d$ Hamilton cycles. A random graph from this class can be constructed (cf. Theorem below)  by picking $d$  Hamilton cycles independently and uniformly at random among all possible Hamilton cycles on the set of $z \geq 3$ vertices, and taking the union of these Hamilton cycles. This construction yields a random regular graph (henceforth called as a {\em random H-graph}) that 
that can be shown to be an expander with high probability (cf. Theorem \ref{th:friedman}). The construction can be accomplished incrementally
as follows.

Let the neighbors of a node $u$ be labeled
as\\
 $nbr(u)_{-1}, nbr(u)_1, nbr(u)_{-2}, . . . , nbr(u)_{-d}, nbr(u)_d$. For each $i$, 
$nbr(u)_{-i}$ and
$nbr(u)_i$ denote a node's predecessor and successor on the $i$th
Hamilton cycle (which will be referred to as the level-$i$ cycle).
We start with 3 nodes, because there is only one possible
H-graph of size 3. 

1.{\bf INSERT(u):} A new node $u$  will be inserted into cycle $i$ between
node $v_i$ and node $nbr(v_i)_i$ for randomly chosen $v_i$, for
$i = 1, \dots, d$. 

2. {\bf DELETE(u):} An existing node $u$ gets deleted by simply removing it
and connecting $nbr(u)_{i}$ and $nbr(u)_{-i}$, for
$i = 1, \dots, d$.

Law and Siu prove the following theorem  (modified here for our purposes) that is used in \Xh:

\begin{theorem}[\cite{lawsiu}]
\label{th:H-graph} Let $H_0, H_1, H_2, \dots$ be a sequence of $H$-graphs, each of size at least 3.
Let $H_0$ be a  random $H$-graph of size $n$ and
let $H_{i+1}$ be formed from $H_i$ by either INSERT or DELETE operation as above.  Then $H_i$ is a random $H$-graph for all $i \geq 0$. 
\end{theorem} 

\begin{theorem}[\cite{friedman, lawsiu}]
\label{th:friedman}
A random $n$-node $2d$-regular $H$-graph is an expander (with edge expansion
$\Omega(d)$)  with probability at least $1 - O(n^{-p})$ where $p$ depends on $d$.
\end{theorem}

Note that in the above theorem, the probability guarantee can be made as close to 1 as possible, by making $d$ large enough. 
Also it is known that $\lambda$, the second smallest eigenvalue, for these random graphs is close to the best possible~\cite{friedman}. Another point to note that although the above construction can yield a multigraph, it can be shown that similar high probabilistic guarantees hold in case we make the multi-edges simple, by making $d$ large enough. Hence we will assume that the constructed expander graphs are simple.
 
We next show how \Xh algorithm is implemented and analyze the time
and message complexity per node deletion.
We note that insertion of a node by adversary involves almost no work from \Xh.
The adversary simply inserts a node and its incident edges (to existing nodes). \Xh simply colors these inserted edges as black. Hence we focus on the steps
taken by \Xh under deletion of a node by the adversary.
 First we state the following lower bound on the amortized message complexity
 for deletions which is easy to see in our model (cf. Section 
 \ref{sec: Xmodel}). Our algorithm's complexity will be within a logarithmic factor of this bound.

\begin{lemma}
\label{le:del}
In the worst case, any healing algorithm needs $\Theta(deg(v))$ messages to repair
upon deletion of a node $v$, where $deg(v)$ is the  degree of $v$ in $G'_t$ (i.e., the black-degree of $v$). Furthermore, if we there are $p$ deletions, $v_1, v_2, \dots, v_p$, then the amortized 
cost is $A(p) = (1/p)\sum_{i=1}^p \Theta(deg(v_i))$ which is the best possible.
\end{lemma}

%\begin{theorem}
%\Xh can be implemented to run in $O(\log n)$ rounds (per deletion).
%The amortized message complexity over $p$ deletions is $O(\kappa\log n A(p))$ on average  where $n$ is the
%number of nodes in the network (at this timestep), $\kappa$ is the degree of the 
%expander used in the construction, and $A(p)$ is defined as in Lemma \ref{le:del}.  
%\end{theorem}

%\begin{proof}
%Proof is deferred to the appendix.
%\end{proof}

\begin{theorem}
Xheal can be implemented to run in $O(\log n)$ rounds (per deletion).
The amortized message complexity over $p$ deletions is $O(\kappa\log n A(p))$ on average  where $n$ is the
number of nodes in the network (at this timestep), $\kappa$ is the degree of the 
expander used in the construction, and $A(p)$ is defined as in Lemma \ref{le:del}.  
\end{theorem}

\begin{proof} (Sketch)
We first note that the healing operations will be initiated by the neighbors of
the deleted node. We also note that primary and secondary expander clouds 
can be identified by the color of their edges (cf. Algorithm in Section\ref{sec:algorithm}.)

{\em Case 1:} This involves constructing a (primary) expander cloud among the neighboring
nodes $N(v)$ of the deleted node $v$. Note that $|N(v)| = deg(v)$, where $deg(v)$ is the black-degree of $v$.    Since  each node knows neighbor of neighbor's (NoN) addresses, it is akin
to working on a complete graph over $N(v)$. We first elect a leader  among   $N(v)$: a {\em random} node (which is useful later) among $N(v)$ is chosen as a leader. This can be done, for example, by using the Nearest Neighbor Tree (NNT) algorithm of \cite{khan-tcs}. This takes $O(\log |N(v)|)$ time and $O(|N(v)|\log |N(v)|)$ messages. The leader  then (locally) constructs a random $\kappa$-regular $H$-graph over $N(v)$  and informs each node in $N(v)$ (directly, since its address is known)
of their respective edges. The total messages needed to inform the nodes is $O(\kappa |N(v)|)$, since that is the total number of edges. A neighbor of the leader in the expander graph is also elected as a vice-leader. This  can be implemented in $O(1)$ time. Hence, overall this case takes $O(\log |N(v)|) = O(\log deg(v)) = O(\log n)$ time and $O(\kappa deg(v)\log deg(v))$ messages.

In particular, the following invariants will be maintained
with respect to  every expander (primary or secondary) cloud: (a) Every node in the cloud
will have a leader (randomly chosen among the nodes) associated with it ; 
%\Amitabh{ In (b), Is it correct to say every node can communicate `directly'?}
(b) every node in the cloud knows the address of the leader and can communicate with it directly (in constant time); and (c) the leader knows the addresses of all other nodes in the cloud; (d) one neighbor of the leader in the cloud will be designated vice-leader which will know everything the leader knows and will take action in case
the leader is deleted. Note that this invariant is maintained in Case 1. We will show 
that it is also maintained in Case 2 below.

Case 2 (Cases 2.1 and 2.2 of Xheal):  We have to implement three main operations in these cases.
They are: \\
{\bf (a)} Reconstructing an expander cloud (primary or secondary) on deletion of a node $v$: %\Amitabh{What is z?}
Let $C$ be the primary (or secondary) cloud that loses $v$.
The node is removed  according to the DELETE operation of $H$-graph.  This takes $O(1)$
time and $O(\kappa)$ messages. If $v$ belongs to $j$ primary clouds then the time is still $O(1)$ while the total message complexity is $O(j\kappa)$. 
For $v$ to belong to $j$ primary clouds its black degree should be at least $j$.
Also $v$ can belong to at most one secondary cloud. Hence the cost is at most
$O(\kappa)$ times the black degree as needed.
 If the deleted node happens to be  the leader of the (primary) cloud then  a new 
{\em random} leader  is chosen (by the vice-leader) and inform the rest of the nodes --- this will take $O(|C|)$ messages and $O(1)$ time, where $|C|$ is the number of nodes in the cloud.
Since the adversary does not know the random choices made by the algorithm,
the probability that it deletes a leader in a step is $1/|C|$ and thus the expected message complexity is $O((1/|C|) |C| = O(1)$. (Note that a new vice-leader, a neighbor of the new leader will be chosen if necessary.)

{\bf (b)} Forming and fixing primary and secondary expander clouds (if there are
enough free nodes): 
Let the deleted node belong to primary clouds $C_1, \dots, C_j$ and possibly
a secondary cloud $F$ that connects a subset of these $j$ clouds (and possibly other unaffected primary clouds). First, each of the clouds are reconstructed
as in (a) above. This operation arises only if we have at least $j$ free nodes,
i.e., nodes that are not associated with any secondary cloud. We now mention
how free nodes are found. To check if there are enough free nodes among the $j$
clouds, we check the respective leaders. The leader always maintain a list of all free nodes in its cloud. Thus if a node becomes non-free during a repair
it informs the leader (in constant time) which removes it from the list. Thus the neighbors of the deleted node can request the leaders of their respective clouds to find free nodes. Hence finding free nodes takes time $O(1)$ and needs $O(j)$ messages. The free nodes are then inserted to
form the secondary cloud. We distinguish two situations with respect
to formation of a secondary cloud:
(i) The secondary cloud is formed for the first time (i.e., a new secondary cloud among the primary clouds). In this case, a leader of one of the associated
primary cloud is elected to construct the secondary expander. This leader
then gets the free nodes from the respective primary clouds, locally constructs
a $\kappa$-regular expander and informs it to the respective free nodes of each primary cloud. This is similar to the construction of a primary cloud as in (a).
The time and message complexity is also bounded as in (a).

(ii)The secondary cloud is already present, merely, a new free node is added.
In this case, the new node is inserted to the secondary cloud by using the INSERT operation of $H$-graph. This takes $O(1)$ time and $O(1)$ messages,
since INSERT can be implemented by querying the leader.

{\bf (c)} Combining many primary expander clouds into one primary expander cloud (if there are not enough free nodes):
This is a costly operation which we seek to amortize over many deletions.
First, we compute the cost of combining clouds. Let $C_1, \dots, C_j$ are the clouds that need to be combined into one cloud $C$.  This is done by first electing a leader  over all the nodes in the clouds $C_1, \dots, C_j$. 
Note that the distance between any two nodes among these clouds is $O(\log n)$,
since all the clouds had a common node (the deleted node) and each cloud is an expander (also note that the neighbors of the deleted nodes maintain
connectivity during the leader election and subsequent repair process). A BFS tree is then constructed subsequently over the nodes of the $j$ clouds with the leader as the root. The leader then collects all the addresses of all the nodes in the
clouds (via the BFS tree) and locally constructs a $H$-graph and broadcasts it to all the other nodes in the cloud. The leader's address is also informed
to all the other nodes in the cloud. Thus the invariant specified in Case 1 is maintained. The total time needed is $O(\log n)$ time and the total number
of messages needed is  $O(\kappa\sum_{i=1}^j |C_i|)\log n$, since each node (other than the leader) sends $O(1)$ number of messages over $O(\log n)$ hops, and the leader sends $O(\sum_{i=1}^j |C_i|)\log n$. However, note that the costly operation of combining is triggered by having less than $j$ free nodes.
This implies that there must been at least $\Omega(\sum_{i=1}^j |C_i|)$
prior deletions that had enough free nodes and hence involved no combining.
Thus, we can amortize the total cost of the combining cost over these ``cheaper'' prior deletions. Hence the amortized cost is
$$\frac{O(\kappa\sum_{i=1}^j |C_i|)\log n }{\Omega(\sum_{i=1}^j |C_i|)} = O(\kappa\log n).$$

Finally, we say how the probabilistic guarantee on the $H$-graph can be maintained.
 The implementation above uses a $\kappa$-regular random $H$-graph in the construction of an expander cloud. By  theorem \ref{th:friedman},  $\kappa$ can be chosen large enough to guarantee the probabilistic requirement needed. For example,
choosing $\kappa = \Theta(\log n)$, then high probability (with respect to the size
of the network) is guaranteed (this assumes that nodes know an upper bound on the size of the network). Furthermore, if there are $f$  deletions, by union bound,
the probability that it is not an expander increases by up to a factor of $f$.
To address this, we reconstruct the $H$-graph after any cloud has lost half of
its nodes; note that the cost of this reconstruction can be amortized over the deletions to obtain the same bounds as claimed.
\end{proof}

%\begin{appendix}

%\section{Proofs}

%\fi

%\subsection{Lower Bounds}
%
%\label{subsec: lowerbounds}
%
%\Amitabh{Fix this. A star graph should have a constant expansion when nodes can only get a constant degree increase}
%
%\begin{theorem}
%Let $n$ be a positive integer, $\alpha \ge 3$ and  $\beta = \frac{1}{2} (\log_{\alpha } (n-1) - 1)$.  Then there exists a
%graph on $n$ vertices and a vertex deletion  such that any way of repairing this deletion
%% any self-healing algorithm in our model 
% under our model must either increase the degree of some node by more than a factor of $\alpha$, or it must increase the distance between some pair of nodes by at least a factor of $\beta$.

%Consider any self-healing algorithm that ensures that: 1) each node increases its degree by a multiplicative factor of 
%at most $\alpha$, where $\alpha \geq 3$; and 2) the stretch of the graph increases by a multiplicative factor of at most
%$\beta$. Then, for some initial graph with $n$ nodes, it must be the case that
% $\beta \geq \frac{1}{2} \log_{\alpha-1}( n - 1)$.
%Then for any positive $\Delta$, for some initial
%graph with maximum degree $\Delta$, it must be the case that
%$\beta  \geq \frac{1}{2} [\log_{\alpha} n - 1]$.
%\end{theorem}

\section{Conclusion}

% Conclusion and future work.

We have presented an efficient,  distributed algorithm  that withstands repeated adversarial node insertions and deletions by adding a small number of new edges after each deletion. It maintains key global invariants of the network while doing only localized changes and using only local information. The global invariants it maintains are as follows. Firstly, assuming  the initial network was connected, the network stays connected. Secondly, the (edge) expansion of the network is at least as good as the expansion would have been without any adversarial  deletion, or is at least a constant. Thirdly,  the distance between any pair of nodes never increases by more than a $O(\log n)$ multiplicative factor than what the distance would be without the adversarial deletions.  Lastly, the above global invariants are achieved while not allowing the degree of any node to increase by more than a small multiplicative factor.

The work can be improved  in several ways in similar models. Can we improve the present algorithm to allow smaller messages and lower congestion? Can we efficiently find new routes to replace the routes damaged by the deletions? Can we  design self-healing algorithms that are also load balanced? Can we reach a theoretical characterization of  what network properties are amenable to self-healing, especially, global properties which can be maintained by local changes? What about combinations of desired network invariants? We can also extend the work to different models and domains. We can look at designing algorithms for less flexible networks such as sensor networks, explore healing with non-local edges. We can also look beyond graphs to rewiring and self-healing circuits where it is gates that fail.

%\iffalse

%\bibliography{selfheal} 
%\bibliographystyle{plain}

%\end{document}

%\endif

%\endif

%\endif

%\bibliographystyle{abbrv}
%\bibliography{selfheal} 

\begin{thebibliography}{10}

\bibitem{anderson01RON}
D.~Andersen, H.~Balakrishnan, F.~Kaashoek, and R.~Morris.
\newblock Resilient overlay networks.
\newblock {\em SIGOPS Oper. Syst. Rev.}, 35(5):131--145, 2001.

\bibitem{garvey}
V.~Arak.
\newblock {What happened on August 16}, August 2007.
\newblock http://heartbeat.skype.com/2007/08/what-happened-on-august-16.html.

\bibitem{AwerbuchAdapt92}
B.~Awerbuch, B.~Patt-Shamir, D.~Peleg, and M.~Saks.
\newblock Adapting to asynchronous dynamic networks (extended abstract).
\newblock In {\em STOC '92: Proceedings of the twenty-fourth annual ACM
  symposium on Theory of computing}, pages 557--570, New York, NY, USA, 1992.
  ACM.

\bibitem{BomanSAS06}
I.~Boman, J.~Saia, C.~T. Abdallah, and E.~Schamiloglu.
\newblock Brief announcement: Self-healing algorithms for reconfigurable
  networks.
\newblock In {\em Symposium on Stabilization, Safety, and Security of
  Distributed Systems(SSS)}, 2006.

\bibitem{chungbook}
F.~Chung.
\newblock {\em Spectral Graph Theory}.
\newblock American Mathematical Society, 1997.

\bibitem{spanders}
S.~Dolev and N.~Tzachar.
\newblock Spanders: distributed spanning expanders.
\newblock In {\em SAC}, pages 1309--1314, 2010.

\bibitem{doverspike94capacity}
R.~D. Doverspike and B.~Wilson.
\newblock Comparison of capacity efficiency of dcs network restoration routing
  techniques.
\newblock {\em J. Network Syst. Manage.}, 2(2), 1994.

\bibitem{fisher}
K.~Fisher.
\newblock {Skype talks of "perfect storm" that caused outage, clarifies blame},
  August 2007.
\newblock
  http://arstechnica.com/news.ars/post/20070821-skype-talks-of-perfect-storm.h%
tml.

\bibitem{friedman}
J.~Friedman.
\newblock On the second eigenvalue and random walks in random d-regular graphs.
\newblock {\em Combinatorica}, 11:331Ð362, 1991.

\bibitem{frisanco97capacity}
T.~Frisanco.
\newblock Optimal spare capacity design for various protection switching
  methods in {ATM} networks.
\newblock In {\em Communications, 1997. ICC 97 Montreal, 'Towards the Knowledge
  Millennium'. 1997 IEEE International Conference on}, volume~1, pages
  293--298, 1997.

\bibitem{mihail-p2p}
C.~Gkantsidis, M.~Mihail, and A.~Saberi.
\newblock Random walks in peer-to-peer networks: Algorithms and evaluation.
\newblock {\em Performance Evaluation}, 63(3):241--263, 2006.

\bibitem{goel04resilient}
S.~Goel, S.~Belardo, and L.~Iwan.
\newblock A resilient network that can operate under duress: To support
  communication between government agencies during crisis situations.
\newblock {\em Proceedings of the 37th Hawaii International Conference on
  System Sciences}, 0-7695-2056-1/04:1--11, 2004.

\bibitem{hayashi2005}
Y.~Hayashi and T.~Miyazaki.
\newblock Emergent rewirings for cascades on correlated networks.
\newblock cond-mat/0503615, 2005.

\bibitem{HayesPODC08}
T.~Hayes, N.~Rustagi, J.~Saia, and A.~Trehan.
\newblock The forgiving tree: a self-healing distributed data structure.
\newblock In {\em PODC '08: Proceedings of the twenty-seventh ACM symposium on
  Principles of distributed computing}, pages 203--212, New York, NY, USA,
  2008. ACM.

\bibitem{HayesPODC09}
T.~P. Hayes, J.~Saia, and A.~Trehan.
\newblock The forgiving graph: a distributed data structure for low stretch
  under adversarial attack.
\newblock In {\em PODC '09: Proceedings of the 28th ACM symposium on Principles
  of distributed computing}, pages 121--130, New York, NY, USA, 2009. ACM.

\bibitem{HillelPODC10}
K.~C. Hillel and H.~Shachnai.
\newblock Partial information spreading with application to distributed maximum
  coverage.
\newblock In {\em PODC '10: Proceedings of the 28th ACM symposium on Principles
  of distributed computing}, New York, NY, USA, 2010. ACM.

\bibitem{holme-2002-65}
P.~Holme and B.~J. Kim.
\newblock Vertex overload breakdown in evolving networks.
\newblock {\em Physical Review E}, 65:066109, 2002.

\bibitem{Wigderson-exsurvey}
S.~Hoory, N.~Linial, and A.~Wigderson.
\newblock {Expander graphs and their applications}.
\newblock {\em Bulletin of the American Mathematical Society}, 43(04):439--562,
  August 2006.

\bibitem{iraschko98capacity}
R.~R. Iraschko, M.~H. MacGregor, and W.~D. Grover.
\newblock Optimal capacity placement for path restoration in {STM} or {ATM}
  mesh-survivable networks.
\newblock {\em IEEE/ACM Trans. Netw.}, 6(3):325--336, 1998.

\bibitem{lawsiu}
C.~Law and K.~Y. Siu.
\newblock {Distributed construction of random expander networks}.
\newblock In {\em INFOCOM 2003. Twenty-Second Annual Joint Conference of the
  IEEE Computer and Communications Societies. IEEE}, volume~3, pages 2133--2143
  vol.3, 2003.

\bibitem{malik}
O.~Malik.
\newblock {Does Skype Outage Expose P2PÕs Limitations?}, August 2007.
\newblock http://gigaom.com/2007/08/16/skype-outage.

\bibitem{medard99redundant}
M.~Medard, S.~G. Finn, and R.~A. Barry.
\newblock Redundant trees for preplanned recovery in arbitrary vertex-redundant
  or edge-redundant graphs.
\newblock {\em IEEE/ACM Transactions on Networking}, 7(5):641--652, 1999.

\bibitem{moore}
M.~Moore.
\newblock {Skype's outage not a hang-up for user base}, August 2007.
\newblock
  http://www.usatoday.com/tech/wireless/phones/2007-08-24-skype-outage-effects%
-N.htm.

\bibitem{motter-2004-93}
A.~E. Motter.
\newblock Cascade control and defense in complex networks.
\newblock {\em Physical Review Letters}, 93:098701, 2004.

\bibitem{motter-2002-66}
A.~E. Motter and Y.-C. Lai.
\newblock Cascade-based attacks on complex networks.
\newblock {\em Physical Review E}, 66:065102, 2002.

\bibitem{murakami97comparative}
K.~Murakami and H.~S. Kim.
\newblock Comparative study on restoration schemes of survivable {ATM}
  networks.
\newblock In {\em INFOCOM}, pages 345--352, 1997.

\bibitem{peleg}
D.~Peleg.
\newblock {\em Distributed Computing: A Locality Sensitive Approach}.
\newblock SIAM, 2000.

\bibitem{ray}
B.~Ray.
\newblock {Skype hangs up on users}, August 2007.
\newblock http://www.theregister.co.uk/2007/08/16/skype\_down/.

\bibitem{SaiaTrehanIPDPS08}
J.~Saia and A.~Trehan.
\newblock Picking up the pieces: Self-healing in reconfigurable networks.
\newblock In {\em IPDPS. 22nd IEEE International Symposium on Parallel and
  Distributed Processing.}, pages 1--12. IEEE, April 2008.

\bibitem{stone}
B.~Stone.
\newblock {Skype: Microsoft Update Took Us Down}, August 2007.
\newblock
  http://bits.blogs.nytimes.com/2007/08/20/skype-microsoft-update-took-us-down.

\bibitem{Amitabh-2010-PhdThesis}
A.~Trehan.
\newblock {\em {Algorithms for self-healing networks}}.
\newblock Dissertation, {University of New Mexico}, 2010.

\bibitem{caenegem97capacity}
B.~van Caenegem, N.~Wauters, and P.~Demeester.
\newblock Spare capacity assignment for different restoration strategies in
  mesh survivable networks.
\newblock In {\em Communications, 1997. ICC 97 Montreal, 'Towards the Knowledge
  Millennium'. 1997 IEEE International Conference on}, volume~1, pages
  288--292, 1997.

\bibitem{xiong99restore}
Y.~Xiong and L.~G. Mason.
\newblock Restoration strategies and spare capacity requirements in
  self-healing {ATM} networks.
\newblock {\em IEEE/ACM Trans. Netw.}, 7(1):98--110, 1999.

\end{thebibliography}

%\iffalse

\end{document}